\documentclass{article}
\usepackage{spconf,amsmath,graphicx}

\def\L{{\cal L}}

\title{Sampling Large Data on Graphs}
\twoauthors
  {Ilan Shomorony}
   {Cornell University\\
   is256@cornell.edu}
  {A. Salman Avestimehr}
   {University of Southern California\\
   avestimehr@ee.usc.edu}

\usepackage{times}

\usepackage{amsthm}
\usepackage{amsmath}    
\usepackage{amssymb}    
\usepackage{mathrsfs}   
\usepackage{epsfig}     
\usepackage{graphicx}
\usepackage{subfigure}

\usepackage{algorithmic}
\usepackage{algorithm}
\usepackage{multirow}

\usepackage[colorlinks=true,pdfstartview=FitV,linkcolor=black,citecolor=black,urlcolor=black,plainpages=false]{hyperref}

\usepackage{amsfonts}
\usepackage{dsfont}

\usepackage[numbers,sort&compress]{natbib}

\usepackage{enumerate}

\allowdisplaybreaks[4]

\renewcommand{\S}{{\mathcal S}}

\newcommand{\N}{{\mathcal N}}

\newcommand{\R}{{\mathbb R}}

\newcommand{\W}{{\mathcal W}}

\renewcommand{\u}{{\bf u}}
\newcommand{\f}{{\bf f}}
\newcommand{\g}{{\bf g}}
\newcommand{\h}{{\bf h}}
\renewcommand{\L}{{\mathcal L}}
\renewcommand{\sp}{{\rm span}}

\newcommand{\non}{\nonumber \\}

\newcommand{\defi}{\triangleq}

\newcounter{constcount}
\newcommand{\eref}[1]{(\ref{#1})}

\newcounter{numcount}
\newcounter{thmcnt}
  \let\Oldsection\section
  
\renewcommand{\section}{\stepcounter{thmcnt}\Oldsection}

\newtheorem{theorem}{Theorem} 
\newtheorem{lemma}{Lemma} 
\newtheorem{definition}{Definition} 
\newtheorem{cor}{Corollary} 
\newtheorem{conj}{Conjecture} 

\newtheorem{remark}{Remark}[section] 
\newtheorem{question}{Question}[section]

\newcounter{examplecounter}
\newcommand{\aln}[1]{\begin{align*}#1\end{align*}}

\newcommand{\al}[1]{\begin{align}#1\end{align}}

\def\Item$#1${\item $\displaystyle#1$
   \hfill\refstepcounter{equation}(\theequation)}

\newcommand{\bea}{\begin{eqnarray}}
\newcommand{\eea}{\end{eqnarray}}
\newcommand{\beas}{\begin{eqnarray*}}
\newcommand{\eeas}{\end{eqnarray*}}

\begin{document}
\ninept
\maketitle
%
%
%
\begin{abstract}
We consider the problem of sampling from data defined on the nodes of a weighted graph, where the edge weights capture the data correlation structure. As shown recently, using spectral graph theory one can define a cut-off frequency for the bandlimited graph signals that can be reconstructed from a given set of samples (i.e., graph nodes). In this work, we show how this cut-off frequency can be computed exactly. Using this characterization, we provide efficient algorithms for finding the subset of nodes of a given size with the largest cut-off frequency and for finding the smallest subset of nodes with a given cut-off frequency. In addition, we study the performance of random uniform sampling when compared to the centralized optimal sampling provided by the proposed algorithms.
\end{abstract}
\begin{keywords}
Sampling, Graph signal processing, cut-off frequency, spectral graph theory
\end{keywords}

\section{Introduction} \label{intro}

Graphs arise as a natural way to represent large datasets obtained in many practical contexts, such as social, biological, and sensor networks \cite{WeberRobust,GirvanCommunity,ZhuSensor}. 
For a graph $G = (V,E)$, the data can be embedded as scalar or vector-valued labels on the vertices $v \in V$, while the weights $w_e$ of the edges $e \in E$ represent some underlying structure in the data.
As an example, one can think of a graph where each vertex corresponds to a different movie title, and the edge weights represent a measure of similarity between the movies.
In this case, the graph data can be the ratings given by a person to each movie title, and one would expect movies connected by edges with large weights to be given similar scores.

Particularly in big data scenarios, a natural question is how well a given sample of the data points can be used to estimate the remainder of the data.
In other words, is it possible to predict the data point at one vertex by interpolating the data from another set of points?
In the context of the movie ratings data, this can be viewed as the celebrated ``Netflix'' challenge \cite{Netflix}, or more in general as data prediction problems for recommendation systems.
Other applications include 
semi-supervised learning of categorized data \cite{ZhuSemisupervised} and ranking problems \cite{HocheRanking}. 

Intuitively, the reason why this graph data interpolation should be at all possible is that the graph contains information about the underlying data structure; thus, a set of samples together with the graph edge weights should reveal information about the missing data points.
As pointed out in \cite{OrtegaInterpolation}, this can be viewed as assuming that the graph data is \emph{slow-varying} or \emph{smooth} on the graph.
Therefore, analogous to the classical signal processing domain, where a smooth signal (i.e., a signal with a small bandwidth) can be recovered from a small set of samples, smoother \emph{graph signals} should have a higher degree of redundancy in their data, and should be recoverable from a smaller set of samples.
These ideas are part of what motivates the emerging field of signal processing on graphs \cite{OrtegaGraphSP} and, in particular, the graph data sampling theory \cite{OrtegaSampling}.

Classical sampling theory states that a signal with bandwidth $W$ can be recovered if we sample at a rate $2W$.
Therefore, given a sampling rate, one can compute the \emph{cut-off frequency}; i.e., the highest frequency component that a given signal may have so that it is recoverable from the samples, which is known as the Nyquist frequency.
In \cite{OrtegaSampling}, 
the authors seek a similar characterization in the context of graph signals, by using tools from spectral graph theory.
The notion of frequency is introduced via the eigenvalues and eigenvectors of the graph Laplacian.
In order to obtain a sampling theorem for graph signals, they consider two questions:
%
%
%
%
%
%
What is the maximum possible
bandwidth (the cut-off frequency) of a graph signal such that it can be recovered from a given subset of nodes, and conversely, what
is the smallest possible subset of nodes that allows the correct recovery of all signals up to a given bandwidth?

Several works prior to  \cite{OrtegaSampling} already dealt with these questions to some extent.
For example, in \cite{OrtegaDownsampling}, the cut-off frequency is established for bipartite graphs.
For arbitrary graphs, sufficient conditions for unique recoverability from a sampling set are stated in \cite{Pesenson}, and then used to derive a lower bound on the cut-off frequency in \cite{OrtegaInterpolation}.
In \cite{OrtegaSampling}, the authors make significant progress towards establishing a sampling theory for graph signals.
They 
present linear-algebraic necessary and sufficient conditions for a given set of samples to correctly recover signals up to a given bandwidth, which is then used to obtain an increasing sequence of lower bounds on the cut-off frequency of a given sampling set.
The drawback of such a characterization is that it is unclear in general 
whether this method can indeed provide arbitrarily close approximations to the cut-off frequency and, if so, how far in the sequence of lower bounds one needs to go.

In this work, we show that the linear-algebraic conditions from \cite{OrtegaSampling} can be used in a different way, which yields an exact characterization of the cut-off frequency.
This is done in Section~\ref{cutoffsec}.
Then, in Section~\ref{algorithmssec}, we show that this characterization can be used to provide efficient algorithms for finding optimal sampling sets, in two senses.
First, what is the subset of nodes of a given size with the largest cut-off frequency?
Second, what is the smallest subset of nodes with a given cut-off frequency?
  In addition, in Section~\ref{uniformsec}, we study the performance of random uniform sampling when compared to the centralized optimal sampling provided by the proposed algorithms.

\section{Notation and Background}

In this section, we introduce the notation and basic notions of spectral graph theory we will need.
We let $G = (V,E)$ be a simple, undirected graph with $|V| = n$ nodes, and we assign a non-negative weight $w_{i,j}$ to each $(i,j) \in E$.
The degree $d_i$ of a node $i \in V$ is given by $d_i = \sum_{j : (i,j) \in E} w_{i,j}$, and we let $D$ be an $n \times n$ diagonal matrix with $d_i$ as the $(i,i)$ entry.
The adjacency matrix $W$ of the graph is an $n \times n$ matrix with $w_{i,j}$ as the $(i,j)$ entry, and we define the Laplacian matrix as $L = D - W$.
We will also be interested in the normalized adjacency and Laplacian matrices, given by
$\W = D^{-1/2}W D^{-1/2}$ and $\L = D^{-1/2} L D^{-1/2}$ respectively.
Both $L$ and $\L$ are symmetric positive semi-definite matrices, and $\L$ has eigenvalues  $\lambda_1,...,\lambda_n$ that in addition satisfy $0 = \lambda_1 \leq \lambda_2 \leq ... \leq \lambda_n \leq 2$. 
We also let $\{\u_1,\u_2,...,\u_n\}$ be the set of orthonormal eigenvectors of $\L$, and $U$ be an $n \times n$ matrix whose $i$th column is $\u_i$.
We will use $\S$ to denote a subset of the nodes in the graph, and $\S^c = V - \S$ to denote the remaining nodes.
A graph signal is a function $f : V \to \R$, 
which can also viewed as a vector  $\f \in \R^n$ with components indexed by the nodes in $V$.
In addition, we let $\f(\S)$ be the vector in $\R^{|\S|}$ with components $\f(i)$, $i\in \S$.

The eigenvalues and eigenvectors of $\L$ (or $L$) can be interpreted as defining a frequency domain for graph signals on $G$.
Analogous to the classical signal processing setting where the Fourier transform converts a time-signal into the frequency domain, the graph Fourier transform (GFT) converts a graph signal $f$ into the basis $\{\u_1,\u_2,...,\u_n\}$.
More precisely, we let $\tilde \f = U^T \f$ be the GFT of $\f$.
It is known that the eigenvalues indeed provide an intuitive notion of frequencies for the graph signal, where the eigenvectors are the corresponding eigenfunctions.
In fact, a higher eigenvalue corresponds to an eigenvector that, when seen as a graph signal on $G$, presents a faster variation across the edges, or is less smooth \cite{OrtegaGraphSP}.
Therefore, it makes sense to define the bandwidth of a graph signal $\f$ to be the largest eigenvalue $\lambda_i$ for which the component of $\f$ along $\u_i$ is nonzero (i.e., $\f^T \u_i \ne 0$).
We define the Paley-Wiener space as
\al{
PW_\omega(G) = \sp ( \u_i : \lambda_i \leq \omega );
}
i.e., the subspace of $\R^n$ with all $\omega$-bandlimited signals.

\section{Characterizing the Cut-off Frequency} \label{cutoffsec}

In order to define the cut-off frequency of set $\S$, we first need to define the concept of a uniqueness set.
Intuitively, $\S$ should be a uniqueness set for some set $A \subset \R^n$ if, from the samples in $\S$, one can correctly reconstruct all graph signals in $A$.
More precisely, we use the following definition from \cite{Pesenson,OrtegaSampling}:

\begin{definition} \label{uniquenessdef}
A set $\S \subset V$ is called a uniqueness set for $A \subset \R^n$ if, for any $\f,\g \in A$, $\f(\S) = \g(\S)$ implies $\f = \g$.
\end{definition}

We can now define the cut-off frequency.

\begin{definition} \label{cutoffdef}
The cut-off frequency $\omega_c(\S)$ of a set $\S$ is the largest $\omega$ such that $\S$ is a uniqueness set for $PW_\omega(G)$.
\end{definition}

One of the contributions of \cite{OrtegaSampling} is the characterization of when $\S$ is a uniqueness set for $PW_\omega (G)$, or more in general for a linear space $M$.
Let $L_2(\S^c)$ be the space of all vectors in $\R^n$ that are zero at all components corresponding to nodes in $\S$.
In \cite{OrtegaSampling}, the following lemma is proved (for the case where $M = PW_\omega(G)$).

\begin{lemma} \label{uniquenesslem}
$\S$ is a uniqueness set for a linear space $M \subset \R^n$ if and only if 
\al{M \cap L_2(\S^c) = \{\bf 0\}. \label{lem1eq}}
\end{lemma}

\begin{proof}
Suppose that $\h \in M \cap L_2(\S^c)$ with $\h \ne {\bf 0}$. 
Then, for any $\f \in M - \{{\bf 0}\}$, we have $\g = \f + \h \in M$.
But this implies that $\f(\S) = \g(\S)$ and $\f \ne \g$.
By Definition~\ref{uniquenessdef}, $\S$ is not a uniqueness set for $M$.
Conversely, suppose $M \cap L_2(\S^c) = \{\bf 0\}$.
Take any $\f,\g \in M$ with $\f(\S) = \g(\S)$.
Then we must have $\f(\S) - \g(\S) = {\bf 0}$, and $\f - \g \in M \cap L_2(\S^c)$, implying that $\f = \g$.
\end{proof}

In \cite{OrtegaSampling}, the authors utilize the characterization of a uniqueness set given by Lemma~\ref{uniquenesslem} to estimate the cut-off frequency of a set $\S$.
More precisely, they show that $\S$ is uniqueness set for $PW_{\omega}(G)$ for any 
$\omega \leq \Omega_k \defi (\sigma_{1,k})^{1/k}$,
where $\sigma_{1,k}$ denotes the smallest eigenvalue of the reduced matrix $(\L^k)_{\S^c}$, obtained by restricting $\L^k$ to the rows and columns corresponding to nodes in $\S^c$.
Since, as shown in \cite{OrtegaSampling}, $(\sigma_{1,k})^{1/k}$ is increasing in $k$, it provides an increasing sequence of lower bounds on the cut-off frequency $\omega_c(\S)$.   

As it turns out, Lemma~\ref{uniquenesslem} can be used in a different way in order to characterize $\omega_c(\S)$ exactly.
Notice that, from \eref{lem1eq}, $\S$ is a uniqueness set for $PW_\omega(G)$ if and only if $PW_\omega(G) \cap L_2(\S^c) = \{\bf 0\}$.
Now, since $L_2(\S^c) = \sp\{ {\bf e}_j : j \in \S^c\}$, where ${\bf e}_j$ is the $j$th standard basis vector, characterizing the largest $\lambda_i$ for which \eref{lem1eq} holds with $M= PW_{\lambda_i}(G)$ can be done by simply testing, for $i=1,...,n$, whether 
\al{ \label{interseceq1}
\sp (\u_1,...,\u_i) \cap \sp( {\bf e}_j : j \in \S^c ) = \{\bf 0\}.
}
This can in fact be done easily for each $i$ by noticing that
\al{
\dim & \left( \sp (\u_1,...,\u_i) \cap \sp( {\bf e}_j : j \in \S^c ) \right) \non
& = i + |\S^c| - \dim \sp(\u_1,...,\u_i, {\bf e}_j : j \in \S^c )  \nonumber \\
& = \dim  \N [ \u_1,...,\u_i, {\bf e}_j : j \in \S^c ], 
\label{dimeq}
}
which implies that \eref{interseceq1} holds if and only if the matrix $[ \u_1,...,\u_i, {\bf e}_j : j \in \S^c ]$ is full column rank.
Therefore, the cut-off frequency $\omega_c(\S)$ can be calculated exactly as described above and we have the following result:

\begin{theorem} \label{cutoffthm}
For a graph $G$ with normalized Laplacian $\L$ with eigenvalues $0 = \lambda_1 \leq \lambda_2 \leq ... \leq \lambda_n$ and corresponding eigenvectors $\u_1,...,\u_n$, 
the cut-off frequency of a subset of nodes $\S$ is given by
\aln{
\omega_c (\S) = \max \left\{ \lambda_i : \dim \N [ \u_1,...,\u_i, {\bf e}_j : j \in \S^c ] = 0 \right\}.
}
Hence $\S$ is a uniqueness set for $PW_\omega(G)$ if and only if $\omega \leq \omega_c(\S)$.
\end{theorem}

The advantage of computing the cut-off frequency using Theorem~\ref{cutoffthm} in comparison to the previously known estimate is illustrated in Fig.~\ref{cutofffig}.
We randomly generated a $300$-node graph by adding each edge with probability $0.4$ and choosing the weight of each existing edge independently and uniformly at random from $(0,1)$.
We then selected a set $\S$ with $30$ nodes at random, and compared $\omega_c(\S)$  to the lower bound given by $\Omega_k = (\sigma_{1,k})^{1/k}$ for increasing values of $k$.
As shown in Fig.~\ref{cutofffig1}, the lower bound does seem to converge to $\omega_c(\S)$ but it seems to require large values of $k$ to be arbitrarily close.
\begin{figure}[ht] 
     \centering
          \subfigure[]{
       \includegraphics[trim=3.7cm 7.3cm 3.9cm 7.5cm,clip=true,width=0.465\linewidth]{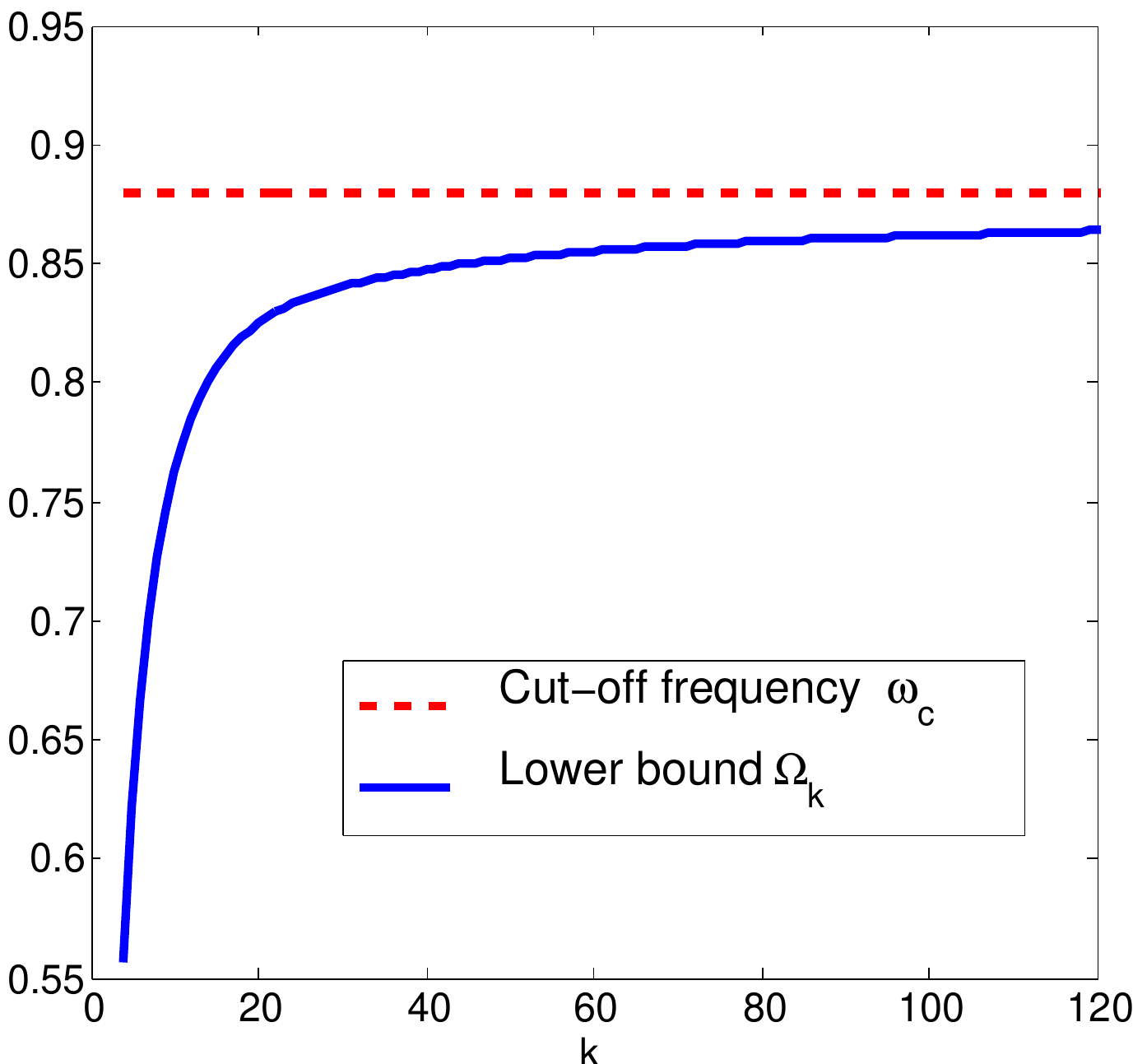}\label{cutofffig1}}
       \hspace{0mm}
                \subfigure[]{
       \includegraphics[trim=3.7cm 7.3cm 3.9cm 7.5cm,clip=true,width=0.465\linewidth]{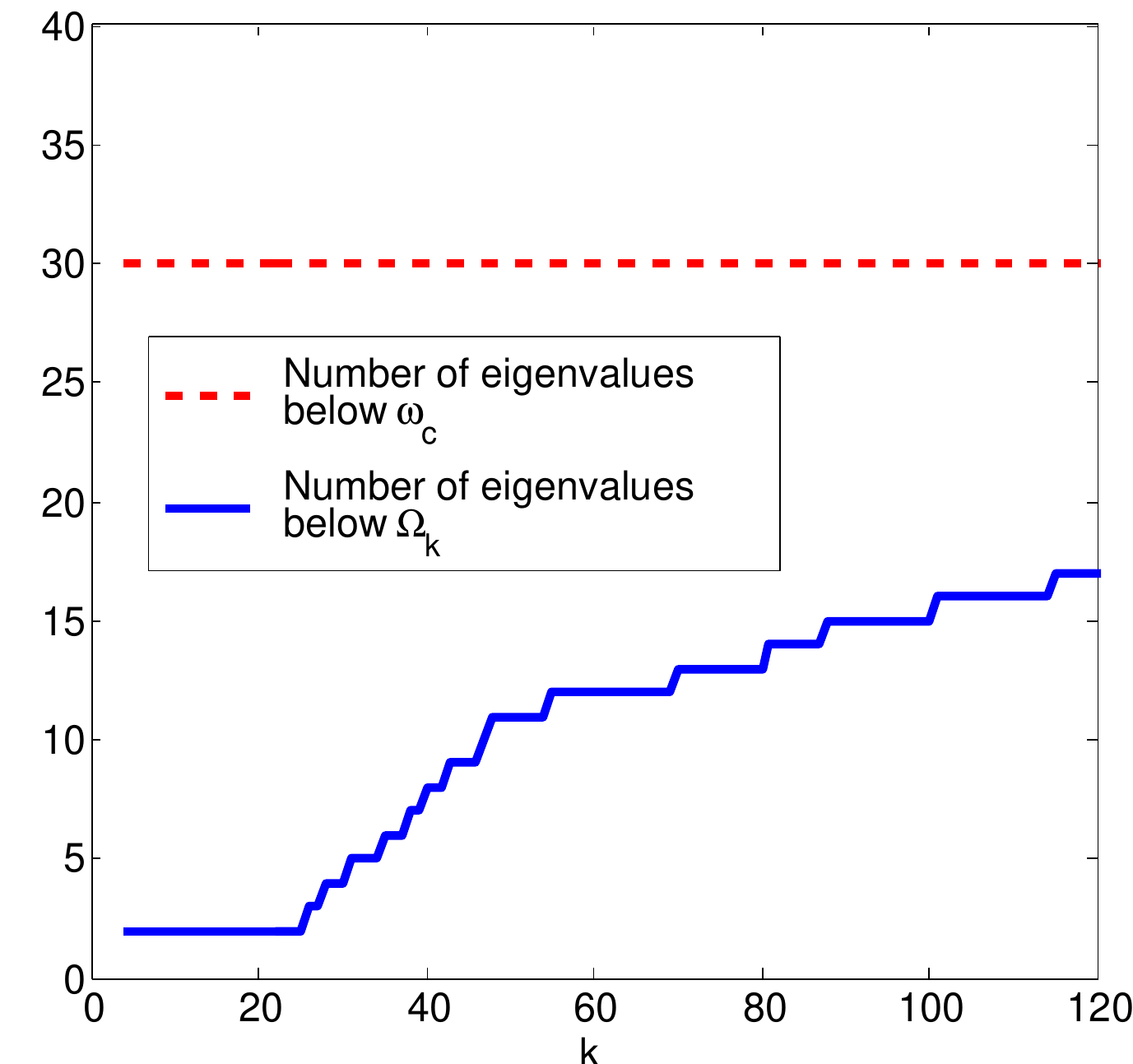} \label{cutofffig2}} 
        \caption{(a) Cut-off frequency for a random set $\S$ and the previously known lower bound. (b)~Number of eigenvalues below the cut-off frequency $\omega_c(\S)$ and below the lower bound $\Omega_k = (\sigma_{1,k})^{1/k}$.\label{cutofffig}}
\end{figure}
In addition, we point out that more important than the actual value of $\omega_c(\S)$ is the number of eigenvalues of $\L$ below $\omega_c(\S)$.
That corresponds to the dimension of the subspace $PW_{\omega_c(\S)}(G)$, which is the set of graph signals that can be correctly reconstructed from $\S$.
While for $k=120$, the approximation given by $\Omega_k$ to $\omega_c$ seems to be good, as shown in Fig.~\ref{cutofffig2}, it implies that $\S$ can reconstruct signals in a subspace of dimension $17$, as opposed to $30$.
Therefore, if we use the true cut-off frequency value as opposed to its estimate in an interpolation technique such as the one described in \cite{OrtegaInterpolation}, a better prediction of the missing data can be obtained.
Finally, we notice that since $|\S| = 30$, by dimensionality considerations we cannot expect $\S$ to reconstruct signals in a space with dimension larger than $30$.
Hence, $\S$ is optimal in the sense of having maximum cut-off frequency, even though it was chosen at random.
As we discuss in Section~\ref{uniformsec}, this seems to be the expected behavior, provided that the graph is connected.

\section{Finding Optimal Sampling Sets} \label{algorithmssec}

Besides characterizing the cut-off frequency of a set $\S$, the approach from the previous section can be used to answer
 two optimization questions related to
 finding optimal sampling sets.
Notice that finding an optimal sampling set, i.e., a set $\S$ with the highest cut-off frequency  under some constraint, has significant practical relevance, since in big datasets,  we are often interested in finding a small yet representative sampling set.
The following two results and their proofs can be understood as providing approaches to selecting optimal sampling sets from the point of view of their cut-off frequencies.
 
 The first problem we consider is to find, for a given $\omega$, the smallest set $\S$ with $\omega_c (\S) \geq \omega$.
 

\begin{cor} \label{cor1}
For a graph $G$ with normalized Laplacian $\L$ with eigenvalues $0 = \lambda_1 \leq \lambda_2 \leq ... \leq \lambda_n$ and corresponding eigenvectors $\u_1,...,\u_n$, 
the problem
\aln{
\min_{\S} \; |\S| \quad \text{subject to} \,\quad \omega_c(\S) \geq \omega, 
}
can be solved in polynomial time and an optimal $\S$ has size $|\S| = m$, where $\lambda_m$ is the smallest eigenvalue of $\L$ such that $\lambda_m \geq \omega$.
\end{cor}

\begin{proof}
Let $\lambda_m$ be the smallest eigenvalue of $\L$ such that $\lambda_m \geq \omega$.
Then, since the set $\{\u_1,...,\u_m\}$ is linearly independent, by the Steinitz exchange lemma, we can find vectors ${\bf e}_{j_{1}},...,{\bf e}_{j_{n-m}}$ in the standard basis of $\R^n$ such that $\{\u_1,...,\u_m, {\bf e}_{j_{1}},...,{\bf e}_{j_{n-m}}\}$ is a basis for $\R^n$.
Hence, if we let $\S = \{1,...,n\} - \{j_1,...,j_{n-m}\}$, we have $L_2(\S^c) = \sp({\bf e}_{j_{1}},...,{\bf e}_{j_{n-m}})$, and by following \eref{dimeq},
\aln{
\dim & \left( \sp (\u_1,...,\u_m) \cap \sp( {\bf e}_{j_{1}},...,{\bf e}_{j_{n-m}} ) \right) \non
& = m + |\S^c| - \dim \sp(\u_1,...,\u_m, {\bf e}_{j_{1}},...,{\bf e}_{j_{n-m}} )  \non
& = \dim  \N [ \u_1,...,\u_m, {\bf e}_{j_{1}},...,{\bf e}_{j_{n-m}} ] = 0.
}
From Theorem~\ref{cutoffthm}, we conclude that 
$w_c(\S) \geq \lambda_m \geq \omega$.
Moreover, for any $\S'$ with $|\S'| < m$, we will have $|(\S^{\prime})^c| > n - m$, and we must have
\aln{
 \dim \N [ \u_1,...,\u_m, {\bf e}_j : j \in (\S^{\prime})^c ] \geq 1,
}
and Theorem~\ref{cutoffthm} now implies that $\omega_c(\S') \leq \lambda_{m-1} < \omega$.

Computing a set $\S$ of minimum size satisfying $\omega_c(\S)\geq \omega$ requires first performing the eigendecomposition of $\L$, and then constructing the basis 
$[ \u_1,...,\u_m, {\bf e}_{j_{1}},...,{\bf e}_{j_{n-m}} ]$ as described in Algorithm 1, all of which can be done in polynomial time, since $\L$ is positive semidefinite. 
\end{proof}


\begin{algorithm}[htb] 
\caption{Computing minimal $\S$ with $\omega_c(\S) = \lambda_m$}
\begin{algorithmic}
\STATE $\S \gets \emptyset$
\STATE $[{\bf b}_1 \, \cdots \, {\bf b}_n] \gets [{\bf e}_1 \, \cdots \, {\bf e}_n]$
\FOR {$\u = \u_1,\u_2,...,\u_m$}
	\STATE \text{Write $\u$ as $\u = \sum_{i=1}^n \alpha_i {\bf b}_i$}
	\STATE \text{$\ell \gets \arg \max_{i \notin \S} |\alpha_i|$}
	\STATE \text{$\S \gets \S \cup \{\ell\}$}
	\STATE \text{${\bf b}_\ell \gets \u$}
\ENDFOR
\end{algorithmic}
\end{algorithm}

The second optimization question is to find, for a given size, the sampling set $\S$ with the maximum cut-off frequency.
As it turns out, the same algorithm provides an efficient solution to this problem.

\begin{cor} \label{cor2}
For a graph $G$ with normalized Laplacian $\L$ with eigenvalues $0 = \lambda_1 \leq \lambda_2 \leq ... \leq \lambda_n$ and corresponding eigenvectors $\u_1,...,\u_n$, 
the problem
\aln{
\max_{\S} \; \omega_c(\S) \quad \text{subject to} \, |S| \leq m, 
}
can be solved in polynomial time and the optimal $\S$ has $\omega_c(\S) = \lambda_m$.
\end{cor}

\begin{proof}
We know from the previous corollary that we can find a set $\S$ of size $|\S| = m$ and cut-off frequency $\omega_c(\S) = \lambda_m$ using Algorithm 1 in polynomial time.
Moreover, for any $\S$ with $|\S| \leq m$, we have $|\S^c| \geq n-m$, and Theorem~\ref{cutoffthm} implies that 
$\omega_c(\S') \leq \lambda_m$.
\end{proof}

In Fig.~\ref{optsetfig}, we illustrate the application of Algorithm 1 to find the optimal set $\S$ in two scenarios.
First we consider a random graph with $200$ nodes on the plane, where edges are added between nodes whose distance is below a fixed threshold and all edges have weight $1$.
In Fig.~\ref{optsetfig1}, we see the optimal set $\S$ with $|\S| = 25$.
As intuition would suggest, the nodes in $\S$ try to cover the graph evenly, and the number of nodes in each connected component seems proportional to its size.
\begin{figure}[ht] 
     \centering
          \subfigure[]{
       \includegraphics[trim=6.5cm 9.5cm 6.5cm 9.5cm,clip=true,width=0.465\linewidth]{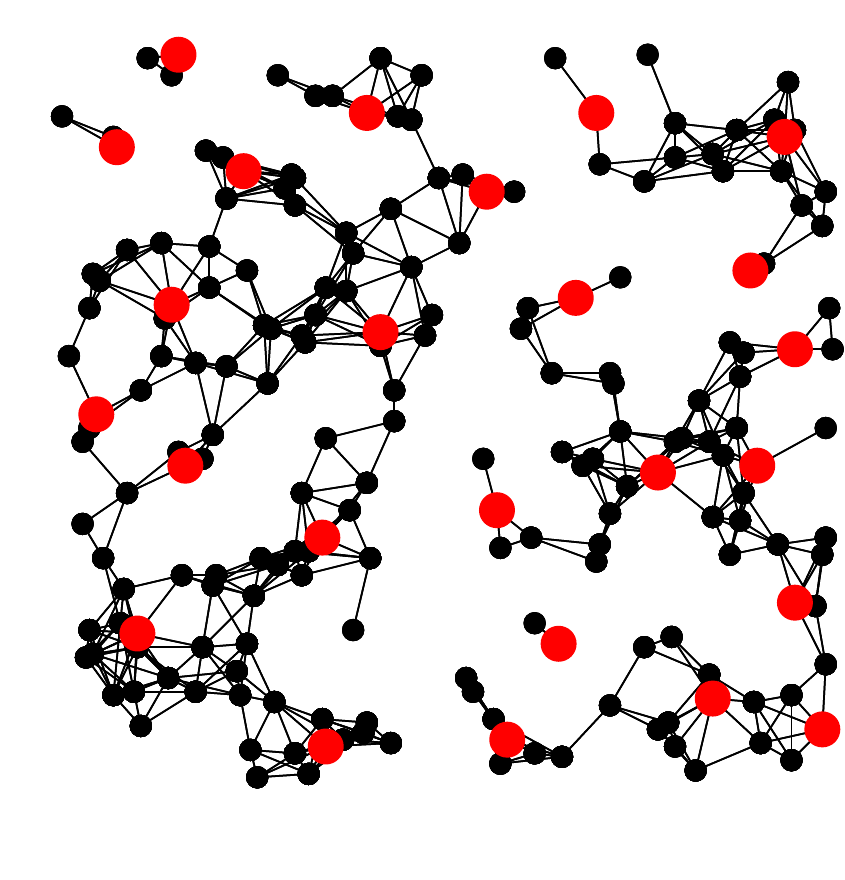}\label{optsetfig1}}
       \hspace{0mm}
                \subfigure[]{
       \includegraphics[trim=5.5cm 8.5cm 5.5cm 8.5cm,clip=true,width=0.465\linewidth]{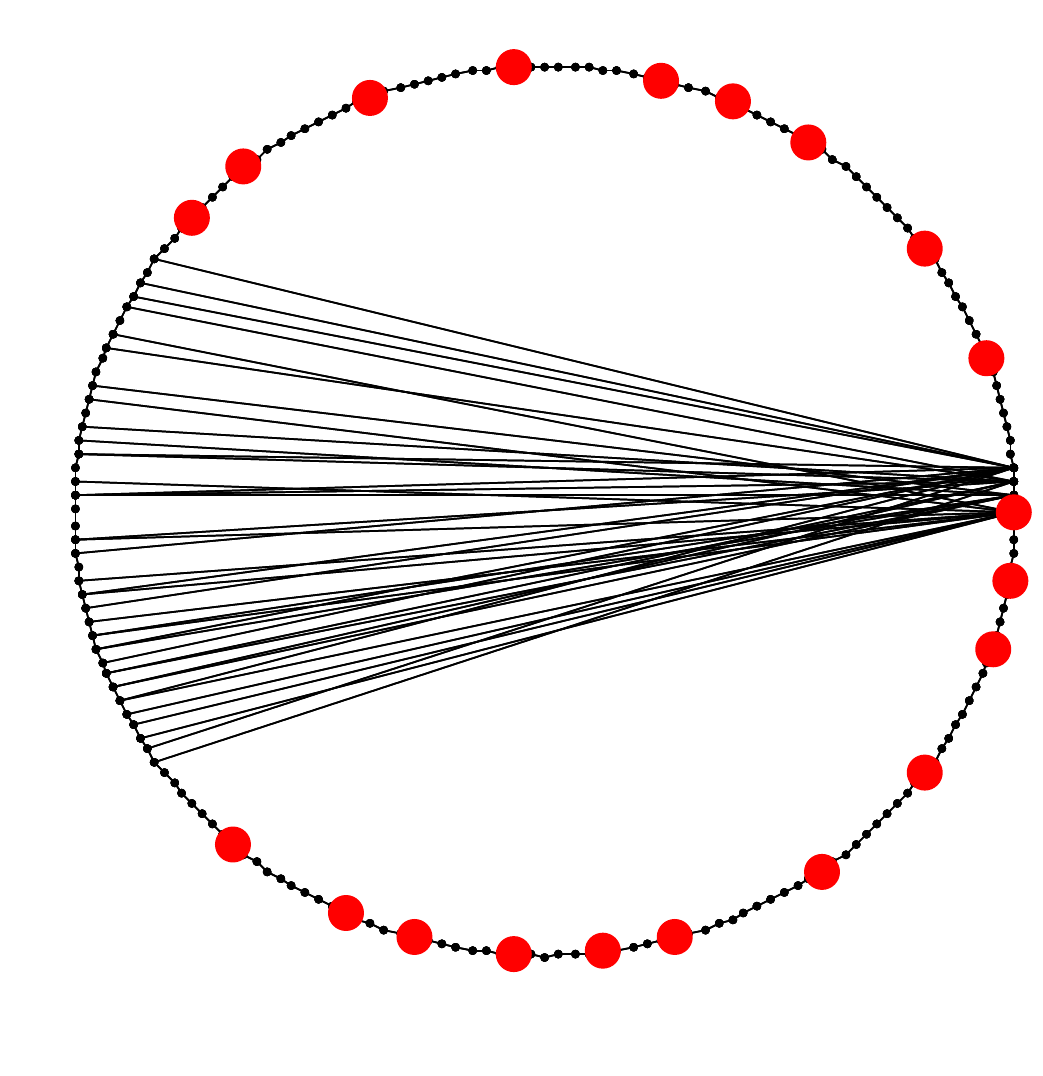} \label{optsetfig2}} 
        \caption{(a) Optimal set $\S$ with $|\S|=25$ (nodes in red) for graph with $200$ nodes. (b)~Optimal set $\S$ with $|\S|=20$ (nodes in red) for $200$-cycle with additional edges.\label{optsetfig}}
\end{figure}
In Fig.~\ref{optsetfig2}, we consider a cycle with $200$ nodes and additional edges connecting a set $A$ of $4$ consecutive nodes to a set $B$ of $40$ consecutive nodes.
We see that an optimal set $\S$ contains one node in $A$ and is essentially evenly distributed over the nodes in $V - B$, since nodes in $B$ are close to the one node chosen from $A$.

\section{Performance of Random Sampling} \label{uniformsec}

As we noticed in Section~\ref{cutoffsec}, for the example illustrated in Fig.~\ref{cutofffig2}, a random set $\S$ of size $|\S| = 30$ has $\omega_c(\S) = \lambda_{30}$.
From Corollary~\ref{cor2}, this is in fact an optimal choice of $\S$ under the constraint $|\S| \leq 30$.
Theorem~\ref{cutoffthm} in fact suggests that this should be the case under fairly general conditions, since it is reasonable that by picking a set of $n-m$ standard basis vectors ${\bf e}_{j_1},...,{\bf e}_{j_{n-m}}$ at random we will have 
\aln{
\dim \N [ \u_1,...,\u_m, {\bf e}_{j_1},...,{\bf e}_{j_{n-m}} ] = 0.
}
Notice however that, if the graph has disconnected components, random sampling may lead to one of the components not being sampled at all and, as illustrated in the example in Fig.~\ref{optsetfig1}, the optimal sampling set tries to keep the number of samples per connected component proportional to the size of the component.
We conjecture the following:

\begin{conj} \label{conj1}
Consider a connected graph $G = (V,E)$ and an arbitrary set $\S \subset V$ with $|\S| = m$.
For almost all assignments of the edge weights, if we let $0 = \lambda_1 \leq ... \leq \lambda_n$ be the eigenvalues of the normalized Laplacian $\L$,
\aln{
\dim \N [ \u_1,...,\u_m, {\bf e}_{j} : j \in \S^c ] = 0,
}
implying that $\omega_c(\S) = \lambda_m$.
\end{conj}

Numerical experiments where we assign the weights to the edges of a connected graph at random give strong support for this claim.
If true, this shows that, in terms of the cut-off frequency, sampling uniformly at random from the nodes in a graph is optimal.
From a practical point of view this is significant since it would obviate the need for a centralized algorithm such as Algorithm 1 to determine an optimal sampling set.

Nonetheless, this also shows a drawback of choosing a sampling set solely based on the cut-off frequency.
For example, consider the graph in Fig.~\ref{optsetfig3}. 
The left half of the $100$ nodes is densely connected, while the right half is not.
As intuition suggests, the optimal sampling set of size $|\S| =30$ picks many more points from the right half of the graph.
\begin{figure}[ht] 
     \centering
          \subfigure[]{
       \includegraphics[trim=5.5cm 8.5cm 5.5cm 8.5cm,clip=true,width=0.465\linewidth]{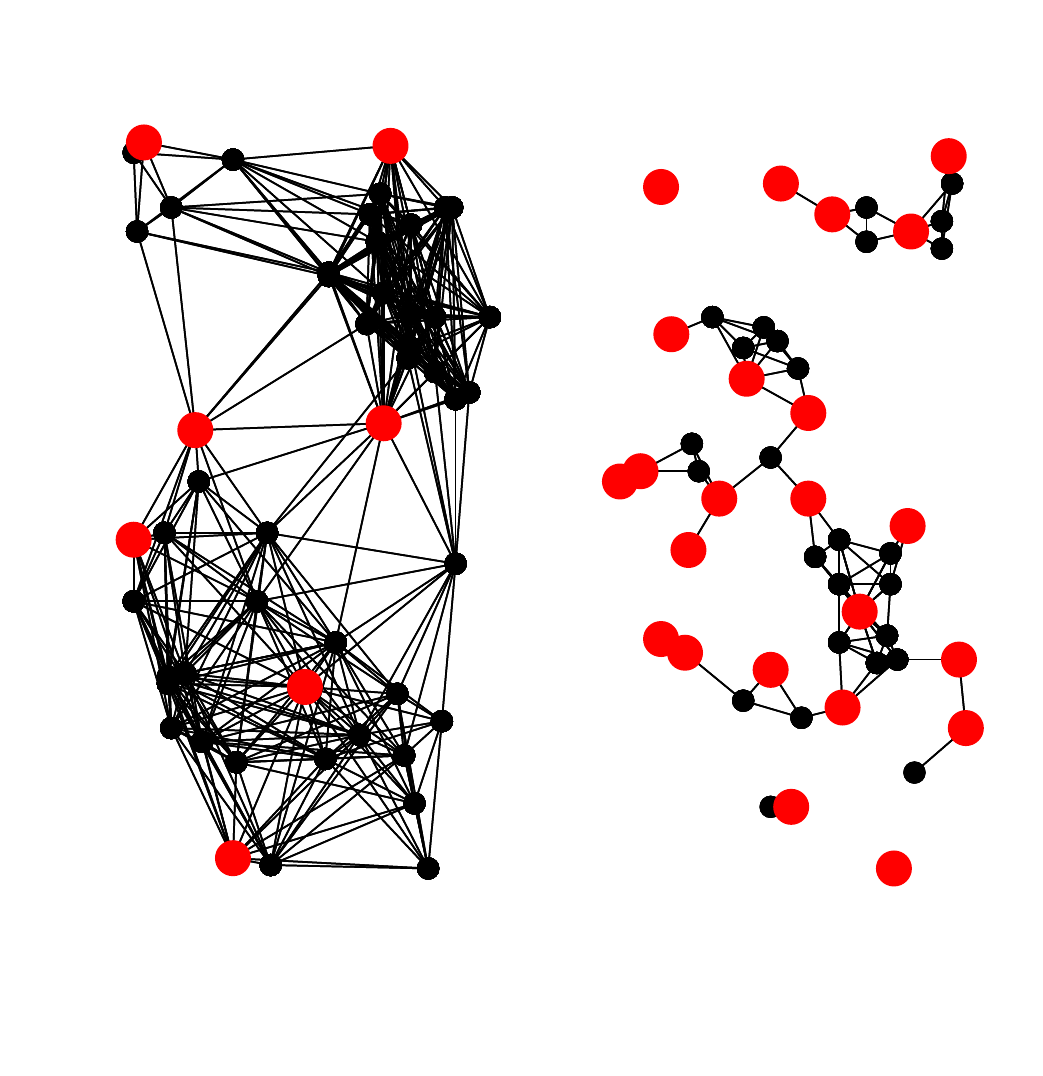}\label{optsetfig3}}
       \hspace{0mm}
                \subfigure[]{
       \includegraphics[trim=5.5cm 8.5cm 5.5cm 8.5cm,clip=true,width=0.465\linewidth]{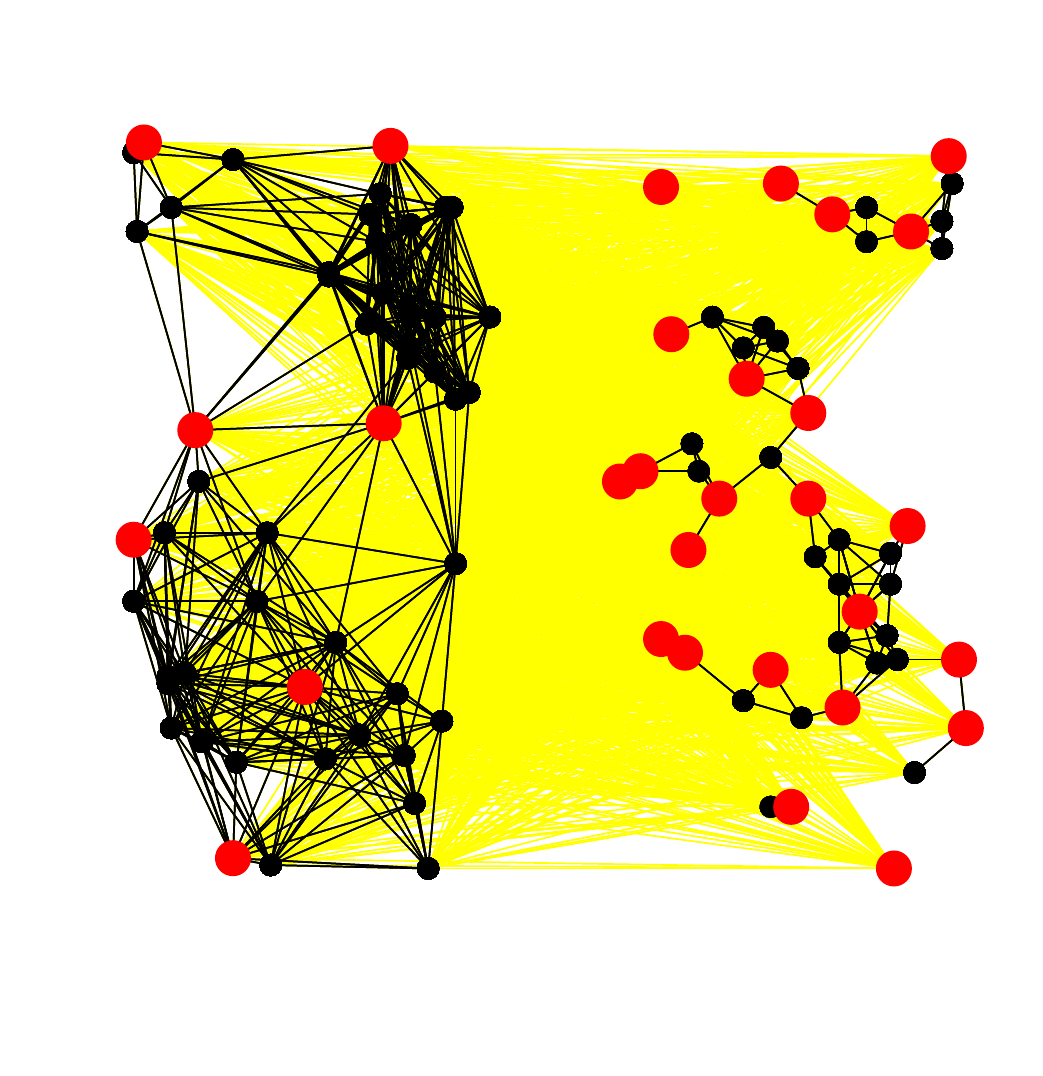} \label{optsetfig4}} 
        \caption{(a) Optimal set $\S$ with $|\S|=30$ (nodes in red) for a graph with $100$ nodes divided into a dense part and a sparse part. (b)~Optimal set $\S$ for the same graph after adding edges of very small weight (shown in yellow) between all pairs of nodes.\label{optsetfig5}}
\end{figure}
In Fig.~\ref{optsetfig4}, we consider adding links of very small weights (shown in yellow) between any two nodes.
For the resulting connected graph, according to Conjecture~\ref{conj1}, any set with $|\S| = 30$ would be optimal from a cut-off frequency point of view.
This suggests that the cut-off frequency is not a robust metric for choosing the best sampling set.
Nonetheless, we point out that in Algorithm 1, by choosing $\ell$ to be $\arg \max_{i\notin \S}|\alpha_i|$ as opposed to any $i$ with $\alpha_i \ne 0$, we try to make sure that the vectors in the resulting basis are ``as orthogonal as possible'' to each other.
This makes the algorithm's output set robust to small variations in the weights and, as shown in Fig.~\ref{optsetfig4}, the optimal sampling set is the same as in Fig.~\ref{optsetfig3}.

\section{Acknowledgements} \label{acksec}

We would like to thank Aamir Anis and Prof.~Antonio Ortega for motivating the problem studied in this paper and for fruitful discussions on the subject.



\vspace{-3mm}

\section{References}
\label{sec:refs}

\vspace{-8mm}

\bibliographystyle{IEEEbib}

\end{document}